\documentclass[ sigconf, nonacm]{acmart}

\usepackage{palatino}
\usepackage{amsthm,amsmath,amsfonts}
\usepackage{graphicx}
\usepackage{hyperref}
\usepackage{booktabs}
\usepackage{algorithm, algpseudocode}
\usepackage{xcolor}
\usepackage{ifdraft}
\usepackage{import}
\usepackage{xspace}
\usepackage[final]{pdfpages}
\usepackage{graphicx}
\usepackage{verbatim}
\usepackage{xfrac}
\usepackage{subcaption}
\usepackage{array}
\usepackage{bbm}
\usepackage{nicefrac}
\usepackage{bm}
\usepackage{mathrsfs}
\usepackage{makecell}
\usepackage{multirow}
\usepackage{apxproof}

\theoremstyle{definition}
\newtheorem{definition}{Definition}
\newtheorem{example}{Example}
\theoremstyle{plain}
\newtheoremrep{theorem}{Theorem}
\newtheoremrep{lemma}{Lemma}
\newtheoremrep{corollary}{Corollary}
\newtheoremrep{proposition}{Proposition}

\DeclareMathOperator{\p}{\mathbb{P}}

\DeclareMathOperator{\Dom}{Dom}
\DeclareMathOperator{\Rows}{Rows}

\newcommand{\js}[1]{\ifdraft{\textcolor{blue}{JS: #1}}\xspace}

\newcommand{\priv}{per-record\xspace}

\begin{document}
\title{Privately Answering Queries on Skewed Data via Per Record Differential Privacy}

\author{Jeremy Seeman$^{\star}$, William Sexton$^{\star}$, David Pujol$^{\star}$, Ashwin Machanavajjhala$^{\star}$}
\affiliation{%
  \institution{$^{\star}$ Tumult Labs, Durham, NC}
}
\email{{jeremy.seeman, william.sexton, david.pujol, ashwin}@tmlt.io}

\begin{abstract}
    We consider the problem of the private release of statistics (like aggregate payrolls) where it is critical to preserve the contribution made by a small number of outlying large entities. We propose a privacy formalism, per-record zero concentrated differential privacy (PRzCDP), where the privacy loss associated with each record is a public function of that record's value. Unlike other formalisms which provide different privacy losses to different records \cite{ebadi_differential_2015,jorgensen_2015_conservative}, PRzCDP's privacy loss depends explicitly on the confidential data. We define our formalism, derive its properties, and propose mechanisms which satisfy PRzCDP that are uniquely suited to publishing skewed or heavy-tailed statistics, where a small number of records contribute substantially to query answers. This targeted relaxation helps overcome the difficulties of applying standard DP to these data products.
\end{abstract}

\maketitle

\section{Introduction} 
We consider the problem of releasing private aggregate statistics on data with highly skewed attributes. These kinds of data occur frequently in practice, for example the Census Bureau's County Business Patterns (CBP) dataset \cite{cbp}, USDA's Census of Agriculture \cite{agcensus}, and many IRS data products like the Corporation Sourcebook \cite{irstaxstats}. Moreover, the aggregate statistics are highly sensitive to contributions by a single (or small set of) units. For example, CBP takes business establishments as its unit of analysis, where it is common that one establishment (like a large retailer or hospital) contributes to the majority of the jobs in a rural area. Regardless of their large contributions, the privacy of these units is often protected under federal law \cite{title13,title26}. Current disclosure avoidance methods, both traditional and modern, either fail to provide strong privacy guarantees or high utility. Classical statistical disclosure limitation techniques like complementary cell suppression using the p\% rule and EZS noise \cite{OMalley2007PracticalCI, evans1996using} offer no formal privacy guarantee and can deterministically reveal information about large contributions. Differentially private (DP) techniques \cite{dwork_calibrating_2006,dwork_algorithmic_2013} that globally bound the contribution of any one unit to published statistics require that highly skewed data are either truncated or suppressed, resulting in unreasonably large bias or unreasonably large noise injected into published statistics. 

Table \ref{tab:intro_examples} gives one such data example. Suppose we wanted to execute the following SQL query using DP: 
\[\mbox{\small{SELECT SUM(Employees) FROM table1 WHERE Industry = 'Retail'}}\]
DP requires bounding the contribution of any one record to the summation, enforced by truncating large values to a clamping bound. Setting this bound too high will require unreasonably large additional noise for DP, whereas setting this bound too low will add unreasonably large bias by truncating record 4. In Section \ref{sec:exp_negative}, we further demonstrate how this problem makes global DP methods ill-suited for summations on skewed data.

\begin{table}[t]
    \centering
    \begin{tabular}{|c|c|c|}
        \hline
         ID& Industry & Employees  \\
         \hline
         1 & Retail & 5 \\
         \hline
         2 & Retail & 5 \\
         \hline
         3 & Retail & 10 \\
         \hline
         4 & Retail & 1000 \\
         \hline
         5 & Technology & 10000 \\
         \hline
         6 & Services & 5 \\
         \hline
         7 & Hospitality & 5 \\
         \hline

    \end{tabular}
    \caption{Sample of skewed data containing a subset of a much larger collection of rows.}
    \label{tab:intro_examples}
\end{table}

Problems like these affect many high-sensitivity queries, where mitigating the effect of one record can have exorbitant effects on utility. For example, DP partition selection algorithms \cite{desfontaines_differentially_2022} used in keyset selection for group-by queries can neglect the impact of individual records with important, outlying attributes. In Table \ref{tab:intro_examples}, if we were to execute the query ``SELECT Count(*) FROM Table1 GROUP BY Industry" using DP, the probability of including "Services" and "Technology" would be the same, despite the fact that the majority of employees in the dataset work in Technology. Issues like these stem from inherent connections between DP and robustness \cite{dworkPrivRobust,liu2021differential, asi2023robustness}, where DP techniques cannot successfully answer inherently non-robust queries.

\subsection{Contributions}

In these settings where standard DP tools do not work, we might consider a relaxation of DP that applies different privacy loss bounds to different units. For example, in Table \ref{tab:intro_examples}, we might allow additional privacy loss for record 5 (Technology establishment with 10000 employees) than record 1 (Retail establishment with 5 employees), especially because record 5 is an influential record which describes more individuals than record 1. However, existing techniques restrict how this mapping between records and privacy losses can occur. Personalized DP \cite{jorgensen_2015_conservative} requires that the mapping between units and their privacy losses is public knowledge, but we might require this mapping to non-trivially depend on confidential data. Alternatively, individual DP \cite{ebadi_differential_2015} requires that this mapping is confidential, but this limits our ability to be methodologically transparent or describe how privacy losses differ across units. 

We propose per-record zero-concentrated DP (PRzCDP), which aims to address these problems. We publicly release what we call a "policy function" $P$ that maps each possible (hypothetical) record to a maximum privacy loss that said record could incur. The form of $P$ can depend on the record's value, allowing different privacy losses for outlying units when necessary to maintain reasonable data utility. We next propose an algorithm methodology, called \textit{unit splitting}, which indirectly sets $P$ by executing traditional zCDP \cite{bun_concentrated_2016} algorithms on a preprocessed version of the data where influential records are split into sub-records. Our formalism and mechanisms are ideal for non-interactive query workloads involving highly skewed data; as a result, our approach is a candidate methodology for data products like CBP \cite{cbp-demo}. Our contributions are as follows: 

\begin{itemize}
    \item In Sections \ref{sec:def} and \ref{sec:properties}, we formally define PRzCDP and demonstrate its formal properties, such as sequential and parallel composition.
    \item In Section \ref{sec:unit splitting}, we propose unit splitting, a pre-processing step which allows us to compute DP mechanisms on the output so that the final results satisfy $P$-PRzCDP. 
    \item In Section \ref{sec:experiments}, we apply these techniques to three datasets: one simulated heavy-tailed dataset, one USDA dataset, and CBP synthetic data provided to us by the U.S. Census Bureau. We empirically demonstrate how small changes in privacy loss significantly improve utility for these skewed datasets.
\end{itemize}

\subsection{Related work}
Among the hundreds of existing formal privacy definitions \cite{desfontaines_sok_2020}, many formalisms aim to provide privacy guarantees which differ across units \cite{alaggan_2016_heterogeneous,jorgensen_2015_conservative,ebadi_differential_2015,ghosh_selling_2015} or realized datasets \cite{papernot2018scalable,wang_per-instance_2019}. These are often used in service of broader global DP goals, such as publishing data-dependent privacy guarantees \cite{redberg2021privately} or establishing privacy filters for adaptive composition \cite{feldman_individual_2021,whitehouse2022fully}. Our work differs in a few key areas. First, we consider explicit dependencies between individual privacy loss parameters and confidential record values, relaxing strong assumptions made about this relationship in previous work \cite{alaggan_2016_heterogeneous,jorgensen_2015_conservative,ebadi_differential_2015,ghosh_selling_2015,papernot2018scalable,wang_per-instance_2019,redberg2021privately}. Second, we consider data-dependent privacy guarantees without global bounds on privacy loss; we consider cases where the policy loss can become arbitrarily large for certain records. Such relaxations are necessary to address problems that arise with skewed data, for which global DP guarantees cannot provide reasonable privacy loss and utility simultaneously.
We provide a detailed comparison between PRzCDP and closely related definitions in Section \ref{sec:priv_relation}.

\section{Preliminaries}
\label{sec:prelim}
\subsection{Data Model}
We assume a single table schema $R(a_1, a_2 \dots a_d)$ where $\mathcal{A} = \{a_1, a_2 \dots a_d\}$ denotes the set of attributes $R$. Each attribute in $a_i$ has domain, $\Dom(a_i)$, which need not be finite or bounded.
The full domain of $R$ is $\Dom(R) = \Dom(a_1) \times \Dom(a_2) \times \dots  \Dom(A_d)$.

A database $D$ is an instance of relation $R$. $D$ is a multi-set whose elements are tuples in $\Dom(R)$, i.e. a tuple can be written as $r = (x_1, \ldots, x_d)$ where $x_i \in \Dom(a_i)$. 
The number of tuples in $D$ is denoted as $|D| = n$.
\subsection{Zero-Concentrated Differential Privacy}
Informally, a randomized mechanism $M$ satisfies DP if the output distribution of the mechanism does not change too much with the addition or removal of a single unit's record. We focus on a variant of DP called $\rho$-Zero-Concentrated Differential Privacy (zCDP)\cite{bun_concentrated_2016}, but all the following results could similarly be adapted for $\epsilon$-DP \cite{dwork_calibrating_2006}. This DP formulation bounds the R\'enyi Divergence of output distributions induced by changes in a single record. We first begin by defining neighboring databases.

\begin{definition}[Neighboring Databases]
Two databases $D$ and $D'$ are considered neighboring databases if $D$ and $D'$ differ by adding or removing at most one row. We denote this relationship by $D' \approx D$. 
\end{definition}

We call the maximum change to a function due to the removal or addition of a single row the \textit{sensitivity} of the function. We will often refer to a single row in a database as a ``unit'' and use the terms interchangeably.

\begin{definition}[$\ell_2$-Sensitivity]\label{def:sensitivity}
Given a vector function $f$, the sensitivity of $f$ is $sup_{D' \approx D}{|f(D)-f(D')|_2}$, where $|\cdot|_2$ is the $\ell_2$-norm, and is denoted by $\Delta$.
\end{definition}
From here, we can define the formal notion of privacy under zCDP.
\begin{definition}[Zero-Concentrated Differential Privacy]\label{def:zCDP}
A rand\-om\-ized mechanism $M$ satisfies $\rho$-zCDP for $ \rho \geq 0$ if, for any two neighboring databases $D$ and $D'$ and for all values of $\alpha \in (1, \infty)$:
\begin{equation*}
D_\alpha(M(D) \| M(D')) \leq \rho \alpha 
\end{equation*}
where $D_\alpha ( \cdot \| \cdot )$ is the R\'enyi divergence of order $\alpha$ between two probability distributions. \par 
\end{definition}
zCDP ensures that no unit contributes too much to the final output of the mechanism by bounding the difference with or without their particular record. The parameter $\rho$ is often called the privacy loss and refers to the amount of information that can be learned about any particular individual. A high value of $\rho$ means weaker privacy protection, while a lower value of $\rho$ denotes a stronger privacy protection.  \par

zCDP has a number of properties that are often used to construct more complex mechanisms. These are composability, post-processing invariance, and group privacy. Composition allows for multiple private mechanisms to compose together to create a large mechanism which still satisfies zCDP. 
\begin{theorem}[zCDP Sequential Composition \cite{bun_concentrated_2016}] \label{thrm:zCDP-sequential}
Let $M_1, M_2$ be randomized mechanisms which satisfy $\rho_1-$zCDP and $\rho_2-$zCDP respectively. Then the mechanism $M'(D) = (M_1(D), M_2(D))$ satisfies $(\rho_1 + \rho_2)-$zCDP.
\end{theorem}
Additionally, if a mechanism is run on multiple disjoint sections of the database, private mechanisms compose with no additional privacy loss. 
\begin{theorem}[zCDP Parallel Composition \cite{bun_concentrated_2016}] \label{thrm:zCDP-parallel}
Let $M_1, M_2$ be randomized mechanisms which satisfy $\rho_1-$zCDP and $\rho_2-$zCDP respectively. Let $D_1, D_2$ be two disjoint subsets of a database, $D$. The mechanism $M'(D) = (M_1(D_1), M_2(D_2))$ satisfies  $\max(\rho_1 , \rho_2)-$zCDP.

\end{theorem}
zCDP also allows for arbitrary post-processing without additional privacy loss. 
\begin{theorem}[zCDP Post-processing \cite{bun_concentrated_2016}] \label{thrm:zCDP-postprocessing}
Let $M$ be a randomized mechanism which satisfies $\rho-$ zCDP. Let $M'(D) = f(M(D))$ for some arbitrary function $f$. Then $M'$ satisfies $\rho-$zCDP.
\end{theorem}
Due to the composition and post-processing theorems, most zCDP mechanisms are built out of simple primitive mechanisms such as the Gaussian mechanism \cite{bun_concentrated_2016} which are then post-processed and combined to create more complex mechanisms.
\begin{definition}[Gaussian Mechanism \cite{bun_concentrated_2016}]
\label{def:gaussian_mech}
Let $q$ be a sensitivity $\Delta$ query. Consider the mechanism $M$ that on input $D$ releases a sample from $\mathcal{N}(q(D),\sigma^2)$. Then $M$ satisfies$ \frac{\Delta ^2}{2\sigma^2}$-zCDP.
\end{definition}

Another property of formally private mechanisms is the notion of group privacy. That is, that the mechanism not only protects the unit but also protects arbitrary groups of units with a privacy loss that scales in the size of the group.
\begin{theorem}[zCDP Group Privacy]\label{thrm:zCDP-group}
Let $M$ be a randomized mechanism which satisfies $\rho-$zCDP. Then $M$ guarantees $(k^2\rho)-$zCDP for groups of size $k$. That is, for every set of neighboring databases $D,D'$ differing in up to $k$ entries, and $\alpha \in (1,\infty)$ we have the following.
\begin{equation*}
D_\alpha(M(D) \| M(D')) \leq (k^2\rho)\cdot \alpha 
\end{equation*}
\end{theorem}

While all the results that follow will use zCDP, they still hold in the context of pure $\epsilon$-DP.

\section{Per-Record Differential Privacy}
\label{sec:def}
Here, we introduce Per-Record Zero-Concentrated Differential Privacy (PRzCDP), a relaxation of DP designed for skewed data tasks. 
This takes the form of a privacy guarantee that varies as a function of the record's confidential value; for example, when records are real-valued positive numbers, we can consider privacy loss bounds that grow monotonically as a function of these record values. The key feature of PRzCDP is a \textit{record dependent policy function} which can be publicly released and analytically captures the privacy loss of a hypothetical record.

\begin{definition}[Record-dependent policy function] \label{def:policy}
A record-dependent policy function $P: \mathcal{T} \rightarrow \mathbb{R}_{\geq 0}$ denotes a maximum allowable privacy loss associated with a particular record value $r \in \mathcal{T}$, where $\mathcal{T}$ is the universe of possible records. 
\end{definition}

This record-dependent policy function differs from other individual privacy frameworks in that the parameter value $P(r)$ itself depends on the confidential record values $r$. We allow the functional form of the policy function $P(\cdot)$ to be made public, but the value of the policy function $P(r)$ for any record $r$ in the confidential database cannot be made public. Record-dependent policy functions are inspired by and generalize binary policy functions of \cite{doudalis_one-sided_2017}. Given a policy function,  Per-Record Zero-Concentrated Differential Privacy is defined as follows.

\begin{definition}[$P$-\priv zero-Concentrated DP ($P$-PRzCDP)] \label{def:PRzCDP}
Let $M$ be a randomized algorithm which outputs a random variable $Y$ over a range $(\mathcal{Y}, \mathcal{F}_Y)$, where $\mathcal{F}_Y$ is an appropriately chosen $\sigma$-algebra. $M$ satisfies $P$-\priv zero-concentrated differential privacy ($P$-PRzCDP) iff $\forall D, D' \in \mathcal{D}$:
$$
D \ominus D' = \{ r \} \implies d_\alpha(M(D) || M(D')) \leq \alpha P(r) \quad \forall \alpha \in (1, \infty)
$$
where $\mathcal{D}$ is the input database space and $\ominus$ denotes symmetric difference.
\end{definition}
In this definition, the privacy loss associated with each record scales according to the policy function, as opposed to having equal privacy loss for all records. Traditional zCDP can be stated as a special case of $P$-PRzCDP where the policy function is a constant. 
\begin{example}
    Consider Table~\ref{tab:ex_splitting}(a) and a policy function $P(r) = \rho \left\lceil \frac{r[Employees]}{50} \right\rceil$, where $\rho$ is a privacy parameter. Under this policy function, Establishment $1$ would receive $3\rho$ privacy loss since they have $150$ employees, while Establishment 2 would receive $\rho$ privacy loss since it only has $50$ employees. Establishment $5$ would still receive $\rho$ privacy loss even though it has less than $50$ employees.
\end{example}

\section{Properties of Per-record differential privacy}
\label{sec:properties}
We demonstrate here that PRzCDP satisfies the traditional properties often associated with Differential Privacy, as well as its variants. First, PRzCDP is closed under post-processing, in that any data independent function computed on the output of a mechanism which satisfies PRzCDP also satisfies PRzCDP.

\begin{lemma}[Closure under post-processing]
\label{lem:PRzCDP-post-processing}
    Given $M:\mathcal{T}^* \to \mathcal{Y}$, a $P$-PRzCDP  mechanism $M$, and any function $f$, it is the case that $f\circ M$ is also $P$-PRzCDP. 
\end{lemma}

PRzCDP satisfies \textit{basic adaptive} sequential composition in that two PRzCDP mechanisms with arbitrary policy functions (chosen prior to running any mechanism) compose together to satisfy PRzCDP in combination. 

\begin{lemma}[basic adaptive sequential composition for $P$-PRzCDP]
\label{lem:PRzCDP-seq}
Let $M_1$ satisfy $P_1$-PRzCDP and let $M_2$ satisfy $P_2$-PRzCDP. Then $M_3(D) = M_2 \left( M_1(D), D \right)$ satisfies $\left(P_1(r) + P_2(r) \right)$-PRzCDP.
\end{lemma}

Like in DP, the privacy losses sum when mechanisms are composed together. In PRzCDP, since the privacy loss is encoded into the policy function, this takes the form of the sum of the two policy functions. \par
PRzCDP also satisfies a form of parallel composition. When multiple mechanisms which satisfy $P$-PRzCDP are run on disjoint subsets of the database, the joint result also satisfies $P$-PRzCDP:

\begin{lemma}[Parallel composition for $P$-PRzCDP]
\label{lem:PRzCDP-par}
Define the partition of size $J \in \mathbb{N} \cup \{ \infty \}$: Let $\mathcal{T}$ define a partition over the universe of possible records. That is,
$$
\mathcal{T} \triangleq \bigcup_{j=1}^J C_j, \qquad C_i \cap C_j = \emptyset, i \neq j
$$
Let $\mathcal{D}_j$ be the space of all databases containing only records in $C_j$ for $j \in [J]$. Let $\{ M_j \}_{j=1}^J$ be mechanisms satisfying $P$-PRzCDP for databases $D_j \in \mathcal{D}_j$ for $j \in [J]$, respectively. Then:
$$
M(D) \triangleq \left\{ M_j(D \cap C_j) \mid j \in [J] \right\}
$$
satisfies $P$-PRzCDP. Note that the $M_j$s can depend on their respective $C_j$s, allowing for adaptivity.
\end{lemma}

Note that Lemma $\ref{lem:PRzCDP-par}$ does not require the individual mechanisms $\{ M_j \}^J_{j=1}$ to be the same for all $j \in [J]$, allowing for adaptive parallel composition.

PRzCDP satisfies the group privacy notion as well. A mechanism that satisfies $P$-PRzCDP also protects groups. 

\begin{lemma}[Simple group privacy for $P$-PRzCDP]
\label{lem:PRzCDP-group}Consider a sequence of databases $D_0, \dots, D_{J}$ where $D_0 = D$ and $D_{j} \ominus D_{j-1} = \{ r_j \}$ for $j \in [J]$. Let $M$ be a randomized mechanism satisfying $P$-PRzCDP. Then we have:
\begin{equation}
d_\alpha(M(D_0)||M(D_J)) \leq \alpha J \sum_{j=1}^J P(r_j)
\end{equation}
\end{lemma}

\begin{lemma}[Advanced group privacy for $P$-PRzCDP]
\label{lem:advanced-PRzCDP-par}
Consider a sequence of databases $D_0, \dots, D_{J}$ where $D_0 = D$ and $D_{j} \ominus D_{j-1} = \{ r_j \}$ for $j \in [J]$. Let $M$ be a randomized mechanism satisfying $P$-PRzCDP. Define $r_{(1)}, \dots, r_{(J)}$ such that:
$$
P(r_{(1)}) \geq P(r_{(2)}) \geq \cdots \geq P(r_{(J)})
$$
Then we have:
\begin{equation}
d_\alpha(M(D_0)||M(D_J)) \leq \alpha \inf_{k \in (1, \infty)} \sum_{j=1}^J \frac{k^j}{k - 1} P(r_{(j)})
\end{equation}
\end{lemma}

\subsection{Relation to Other Privacy Formulations}
\label{sec:priv_relation}

Prior work has studied the idea of giving different privacy guarantees to different records, and the idea of accounting privacy loss as a function of the data. The flexibility of PRzCDP, by comparison, lies in the fact that each unit's privacy loss is a \emph{function} of their private record, and only that function is published instead of particular values. Units with knowledge of their private record use this public function to `look up' their privacy loss. 

PRzCDP most closely resembles the definition for Personalized Differential Privacy (PRDP) \cite{jorgensen_2015_conservative}. Like PRzCDP, PRDP gives different guarantees to different records: 

\begin{definition}[Personalized zCDP (PRDP)\cite{jorgensen_2015_conservative}] 
    Let $\mathcal{T}$ be a universe of participating records and $\Phi: \mathcal{T} \mapsto \mathbb{R}^+$ be a function which maps each unit to a privacy loss. A randomized algorithm $M$ satisfies $\Phi$-PRDP if, for any two databases $D, D'$ which differ on the contributions of one unit $r \in \mathcal{T}$, we have

        $$d_\alpha(M(D) || M(D')) \leq \alpha \Phi(r)$$

    \label{def:PRDP}
\end{definition}
This definition contrasts with PRzCDP by assuming that $\Phi(r)$ is public knowledge for all $r \in \mathcal{T}$ \emph{i.e.} each unit's privacy guarantee is publicly publishable. Fundamentally, this requires that the guarantee, $\Phi(r)$, of each unit, $r$, is independent of their private record. By comparison, $P$-PRzCDP does not make any assumptions about independence of a unit's privacy guarantee and their sensitive record. Thus, each unit's guarantee, $P(r)$ for the unit's record $r$, remains confidential. Only the policy function, $P(\cdot)$, is published and individual unit with knowledge of their private record can compute their own guarantee.

Other works have also studied computing privacy loss as a function of the data but only in an effort to give tighter accounting of the global privacy loss, which is constant across all participants and is made public. For instance, \citet{pate} give tight privacy loss accounting by deriving a global loss from the data itself. This global loss is then passed through a novel mechanism to publish a noisy private version. PRzCDP instead gives per-record guarantees that are not published. Similarly, the individualized accounting method of \cite{feldman_individual_2021} computes each record's loss as a function of its data. However, this is only to ensure that no single record exceeds the public global privacy budget, constant across all records. Alternatively, \cite{redberg2021privately} uses an existing global DP guarantee to provide an ex-post characterization for the gap between the global privacy loss and the confidential data-dependent realized privacy loss. The foundational distinguishing factor of PRzCDP from both these approaches is that only the policy is published, as opposed to any global or individual privacy losses. 

DP encompasses a wide variety of formalisms \cite{desfontaines_sok_2020} which rely on alternative characterizations of scenarios under comparison, measures of privacy loss between those scenarios, and the generality of bounds on these measures. Here, we show how PRzCDP is interoperable with these definitions. Connections can be made to one-sided DP \cite{doudalis_one-sided_2017}, also adapted to the semantics of zCDP:                          

\begin{definition}[One-sided zCDP (OSzCDP)\cite{doudalis_one-sided_2017}]
    Let $P: \mathcal{R} \mapsto \{ 0, 1 \}$ be a function that labels records as privacy-sensitive ($P(r) = 0$) or not ($P(r) = 1$). A randomized algorithm $M$ satisfied $(P, \rho)$-one-sided zero-concentrated DP if, for any two databases $D, D'$ where
    \begin{equation}
        D' = D \setminus \{ r \} \cup \{ r' \}, P(r) = 0, r \neq r'    
    \end{equation}
    we have
    \begin{equation}
        d_\alpha(M(D) || M(D')) \leq \alpha \rho
    \end{equation}
\end{definition}

\begin{lemma}
    Suppose there exists $\rho \geq 0$ and a subset of records $R \subseteq \mathcal{T}$ such that:
    \begin{equation}
        \sup_{r \in R} P(r) \leq \rho
    \end{equation}
    Then for the policy $P^*: \mathcal{T} \mapsto \{ 0, 1 \}$ where $P(r) = \mathbbm{1}\{ r \notin R \}$, any mechanism $M$ that is $P$-PRzCDP is also $(P^*, 2\rho)$-OSzCDP.
\end{lemma}

\section{Mechanisms for PRzCDP}
\label{sec:unit splitting}

In this section, we present a novel class of privacy mechanisms for ensuring PRzCDP. This class of mechanisms is called \textit{Unit Splitting}. As the name suggests, the mechanisms follow this general pattern: 
\begin{itemize}
\item Preprocess the input dataframe by ``splitting" each row or record into many smaller rows or sub-records. The number of splits depends on a measure of how large the row is. 
\item Next, we run a mechanism that satisfies standard 
$\rho$-zCDP. 
\item By group composition, the privacy loss of a row or record in the original dataframe will be $k^2\rho$, where $k$ is the number of times an original row is split in the preprocessing step.  
\end{itemize}

This allows a PRzCDP mechanism to be built by doing a pre-processing step followed by any arbitrary zCDP mechanism. We introduce the basics of unit splitting in Section~\ref{sec:unit-splitting-intro}, and describe how to use them to answer SQL aggregation queries and group-by aggregation queries in Sections~\ref{sec:unit-splitting-aggregations} and \ref{sec:unit-splitting-groupby}, respectively. 

\subsection{Unit Splitting}\label{sec:unit-splitting-intro}
Unit splitting is a preprocessing step that uses a mapping function $A(r)$ to map each record into one or more other records. Answering queries on the split records using a mechanism that satisfies zCDP results in an overall mechanism that satisfies PRzCDP. We state this more formally as follows.
\begin{lemma} [$\rho$-zCDP with pre-processing implies $P$-PRzCDP] \label{lem:pre-processing}
    Consider a pre-processing function $A: \mathcal{T} \mapsto \mathcal{T}^*$  where $A$ maps each record $r \in \mathcal{T}$ to a multiset of records in $\mathcal{T}^*$. Let $|A(r)|$ be the cardinality of the multiset $A(r)$, i.e., the number of subsequent records generated by $A(r)$. If $M$ is a $\rho$-zCDP algorithm operating on $\mathcal{D}^*$, then $M(A(D))$ satisfies $P$-PRzCDP where $P(r) = \left( \rho |A(r)|^2 \right)$.
\end{lemma}
\begin{proof}
Let $D$ and $D'$ be neighboring datasets and, without loss of generality, let $A(D') \setminus A(D) = \{ s_1, \dots s_{|A(r)|} \} \subseteq \mathcal{T}^*$. By the group-privacy properties of $\rho$-zCDP:
\begin{align}
    D_\alpha(M(A(D)) || M(A(D'))) \leq \alpha \left( |A(r)|^2 \rho \right)
\end{align}
\end{proof}

This allows a practitioner to create mechanisms which satisfy $P$-PRzCDP by using the unit-splitting preprocessing step followed by an off-the-shelf zCDP mechanism. This applies to all zCDP mechanisms from existing private frameworks such as Tumult Analytics\cite{berghel2022tumult}, to complex mechanisms such as stochastic gradient descent\cite{AbadiCGMMT016} and the matrix mechanism\cite{li_matrix_2015}. The policy function in this case is implied by the choice of unit splitting. The number of splits per record directly impacts the privacy loss of that record, with those that require a higher number of splits receiving a larger privacy loss than those with a smaller number of splits. \par 
Since each record's privacy loss is now a function of the contents of the record, the privacy loss is considered a private value, and cannot be published. Instead, for transparency, the splitting function itself can be released, which would allow an observer to reason about the privacy loss of hypothetical records without releasing information about the records. 

Unit splitting is a powerful tool that can be applied to a wide variety of tasks. In the following sections, we'll demonstrate how to use the unit splitting paradigm for one such task. Specifically, PRzCDP algorithms for answering SQL group-by aggregation queries for highly skewed data.

\begin{table}[t]
    \begin{subtable}{\linewidth}
      \centering
        \begin{tabular}{|c|c|c|c|}
        
            \hline
            ID & Industry & Employees & Payroll   \\ 
            \hline
            1 & Agriculture & 150 & \$ 10,000,000  \\
            \hline
            2 & Agriculture & 50 & \$ 15,000,000  \\
            \hline
            3 & Mining & 100 & \$ 10,000,000  \\
            \hline
            4 & Mining & 50 & \$ 10,000,000  \\
            \hline
            5 & Retail & 20 & \$ 1,000,000  \\
            \hline
        \end{tabular}
        \label{tab:pre_split}
        \caption{Pre-Split Table}
    \end{subtable}%
    \\
    \begin{subtable}{\linewidth}
      \centering
        \begin{tabular}{|c|c|c|c|}
        
            \hline
            ID & Industry & Employees & Payroll  \\
            \hline
            1 & Agriculture & 50 & \$ 5,000,000 \\
            \hline
            1 & Agriculture & 50 & \$ 5,000,000  \\
            \hline
            1 & Agriculture & 50 & \$ 0 \\
            \hline
            2 & Agriculture & 50 & \$ 5,000,000  \\
            \hline
            2 & Agriculture & 0 & \$ 5,000,000  \\
            \hline
            2 & Agriculture & 0 & \$ 5,000,000  \\
            \hline
            3 & Mining & 50 & \$ 5,000,000  \\
            \hline
            3 & Mining & 50 & \$ 5,000,000  \\
            \hline
            4 & Mining & 50 & \$ 5,000,000  \\
            \hline
            4 & Mining & 0 & \$ 5,000,000  \\
            \hline
            5 & Retail & 20 & \$ 1,000,000  \\
            \hline
            
        \end{tabular}
        \label{tab:post_split}
        \caption{Post-Split Table}
    \end{subtable} 
 \caption{Sample table before and after unit splitting. The splitting thresholds used were as follows. Employees: 50, Payroll: \$ 5,000,000 \js{These tables take up a lot of real estate, are they worth compressing?}}
 \label{tab:ex_splitting}
\end{table}

\subsection{Answering Aggregation Queries Using Unit Splitting}\label{sec:unit-splitting-aggregations}

We can privately compute aggregates such as sums on skewed data by applying unit splitting in the form of a splitting threshold. Doing so splits the few large contributors with possibly unbounded values into several smaller bounded values. This method both bounds and reduces the overall sensitivity of many queries, therefore allowing lower-error private answers. It, however, comes at the cost of higher privacy loss for larger records which were split into multiple smaller records. The choice of splitting procedure results in an implicit policy function for PRzCDP. \par 
We make a distinction between conditional attributes and measure attributes. Conditional attributes are those which will be used in conditional statements, such as the WHERE clause in an SQL query. These get duplicated across all splits of the data. Measure attributes are those which are computed in aggregations and thus get split across all the split records.
For numerical measure attributes, we individually evaluate the minimum number of times each of the magnitude attributes $a$ need to be split (i.e., the smallest integer $m$ such that $m T(a) \geq r(a)$). For example, if a row has $r(a)=12$, and the splitting threshold is $T(a) = 5$, then the row would need to be split into at least 3 sub-records, corresponding to  values of $\{ 5, 5, 2\}$. This minimum number of splits may be different for each attribute in any particular row; to resolve this difference, we select the magnitude attribute with the \emph{largest} number of required splits. The remaining split elements are zero-padded. Ideally, each magnitude attribute will be split the same number of times, to prevent too many zero-padded elements. As a complete example, Table \ref{tab:ex_splitting} lists how several records would be split according to these rules. 

\begin{algorithm}[!htbp]
\caption{\label{alg:unit-splitting}Unit Splitting Pre-Processing} 
\begin{algorithmic}[1]
\Require $D$: Private dataframe.
\Require $T: \mathcal{A} \mapsto \mathbb{Z}$: splitting threshold function.
\Ensure $\{ r_i \}$: multiset of unit splitting rows.

\Procedure{UnitSplit}{$D$, $T$}

\For{$r \in \Rows(D)$}
    \State Find the smallest integer $m$ such that $mT(a) \geq r(a)$ for all $a \in \mathcal{A}$
    \State Split $r$ into $m$ rows such that for each split and each attribute $r_i(a) \leq T(a)$ and $\sum_i r_i(a) = r(a)$
    
\EndFor
\EndProcedure
\end{algorithmic}
\end{algorithm}

Once each record has been split, the split records are used to answer each query in the workload.  For each query on the original unsplit records, we define an equivalent rewritten query on the split records. Table \ref{tab:rewritten_queries} describes how each query is rewritten. For example, in order to ensure the non-private results agree between the split and unsplit records, we count the number of distinct IDs for the COUNTs when using the split records.
\begin{lemma}\label{lem:no-bias}
    As $\rho $ tends to $ \infty$, the difference between rewritten queries of Table~\ref{tab:rewritten_queries} and their unsplit counterparts tends to $0$. 
\end{lemma}
\begin{proof}
    By construction, non-private sum queries will be equivalent using either the split or unsplit data. Since the values of the aggregation columns are distributed without loss or additions to each of the split rows, the sum of the split rows is equivalent to the sum of the unsplit rows. Likewise, since the row IDs are copied over to all the split rows, the COUNT\_DISTINCT query will only increment by one for each unsplit row. Since an average of the unsplit rows is the sum divided by count, the equivalent on the split rows is the sum divided by the distinct count.
\end{proof} 
Lemma~\ref{lem:no-bias} simply states that the process of unit splitting itself incurs no bias for summation queries and asymptotically negligible bias for other inexact reconstructed queries.
Since these aggregations are simply a zCDP mechanism after a mapping, by Lemma~\ref{lem:pre-processing}, the entire process satisfies $P$-PRzCDP where the policy function is dependent on the number of times each record is split. 
\begin{lemma}
    Computing Algorithm~\ref{alg:unit-splitting} followed by an aggregation from Table~\ref{tab:rewritten_queries} satisfies $P$-PRzCDP.
\end{lemma}
Conditionals such as group-by and filters can be applied to split data without any adaptation, since the conditional attributes are duplicated across all splits. We give an example of answering a private sum below.
\begin{example}\label{ex:sum}
    Consider taking a sum over the Employees column of Table~\ref{tab:ex_splitting}(a). We use the following splitting threshold (Employees: 50, Payroll: \$5,000,000).  The table after splitting can be found in Table~\ref{tab:ex_splitting}(b). After splitting, the maximum value of the Employees column is 50 and each record is split into multiple rows with at most $50$ employees. The same holds for Payroll and its associated threshold. Once the table is split, the sum is taken over all the split rows.\par
    The PRzCDP policy function is implied by the splitting \linebreak threshold. Each record incurs a privacy loss  according to the number of times it is split. Establishments $1$ and $2$ are split $3$ times and incurs a privacy loss of $9\rho$. Establishments $3$ and $4$ are only split twice and each incurs a privacy loss of $4\rho$. Establishment $5$ is never split and incurs a privacy loss of $\rho$.
\end{example}
In this example, we used a sum, but this could include other aggregations such as averages or other more complex zCDP mechanisms such as partition selection \cite{desfontaines_differentially_2022}, matrix mechanism \cite{li_matrix_2015} among others. One benefit of unit splitting is bounding the sensitivity of previously unbounded queries.
\begin{example}\label{ex:sum2}
    Consider taking a sum over the Employees column of Table~\ref{tab:ex_splitting}(a). Prior to unit splitting, the maximum possible number of employees for any arbitrary establishment is unbounded. There is no limit to the number of employees  an establishment can have. After the splitting, the maximum value of the Employees column is set to 50, introducing a bound. Since the maximum number of employees in any record is 50, the sensitivity of the sum over employees is also 50.
\end{example}
In this case, the previously unbounded sensitivity is now bounded by the unit splitting algorithm. Traditionally, one would set a clamping bound on the sum, which would truncate all the values outside the bound to one of the boundary values. For heavily skewed data, this can either introduce bias (if the bound is too small) or a large amount of noise (if the bound is too large). When using unit splitting, this is no longer a concern, as the splitting threshold can be set low enough to avoid incurring too much noise. The tradeoff in this case is that the larger records incur a larger privacy loss overall. This makes PRzCDP particularly powerful in the case of highly skewed data, where a small amount of large records makes it impractical to introduce clamping bounds. Instead, these records are split into many smaller records and incur a larger privacy loss as a result.

\begin{table*}[t]
    \centering
    \begin{tabular}{|c|c|c|}
        \hline
        Original query & Rewritten Exact Query  & PRzCDP $P(r)$ \\
        \hline 
        \texttt{COUNT(ROW\_ID)} & \texttt{COUNT\_DISTINCT(ROW\_ID)}  & $\rho$ \\
        \hline 
        \texttt{SUM($\cdot$)} & \texttt{SUM($\cdot$)}  & $\rho |A(r)|^2$ \\
        \hline 
        \texttt{AVG($\cdot$)} & \texttt{SUM($\cdot$) / COUNT\_DISTINCT(ROW\_ID)}  & $\rho |A(r)|^2$ \\
        \hline
    \end{tabular}
    \caption{Common exact queries and their reconstruction strategies}
    \label{tab:rewritten_queries}
\end{table*}
\subsection{Multiple Aggregations}
While unit splitting is a versatile technique in its own right, much of its power comes from its ability to compose neatly with much of the existing DP literature as well as other instances of unit splitting. While any two PRzCDP mechanisms compose together due to Lemma~\ref{lem:PRzCDP-seq}, this also holds for mechanisms run prior to the unit splitting. This is because $\rho$-zCDP can be seen as a special case of PRzCDP where the policy function is equal to the privacy parameter $\rho$.
\begin{lemma}\label{lem:PRzCDP_special_case}
    Let $M$ be a randomized mechanism which satisfies $\rho$-zCDP. Then $M$ also satisfies $P$-PRzCDP where $P(r) = \rho$.
\end{lemma}
This allows mechanisms run prior to unit splitting to compose with those run after unit splitting by using Lemma~\ref{lem:PRzCDP-seq}.

PRzCDP excels in reducing the sensitivity of high or unbounded sensitivity queries, such as sums or means. However, when computing queries with low sensitivities such as counts or medians, it is easier to use zCDP to compute these counts prior to unit splitting. In these cases, each record only incurs a constant privacy loss (the $\rho$ parameter given to the zCDP mechanism) as opposed to the variable privacy loss given by PRzCDP and unit splitting. We demonstrate this in the following example where zCDP is used to compute a median followed by a sum computed on unit split data.
\begin{example} \label{ex:median_sum}
    Consider the sum from Example~\ref{ex:sum} with the addition of a median query on the Payroll column prior to the unit splitting process. If we first compute the median with a privacy budget of $\rho_1$, then compute the sum with privacy-loss budget $\rho_2$, then by Lemma~\ref{lem:PRzCDP-seq} the combination of the two mechanisms satisfies PRzCDP with a policy function of $\rho_1 + \rho_2 |A(r)|^2$, where $|A(r)|$ is the maximum number of times record $r$ is split. In this case, Establishments $1$ and $2$ incur $\rho_1 + 9\rho_2$ privacy loss, Establishments $3$ and $4$ incur $\rho_1 + 4\rho_2$ privacy loss, and Establishment $5$ incurs $\rho_1 +\rho_2$ privacy loss.
\end{example}
By using traditional zCDP to compute the median, each record only incurs a constant ($\rho_1$) privacy loss. Had that median been computed after the unit splitting, the larger records would have incurred $9\rho_1$ privacy loss instead, a significant increase.
In these cases, zCDP mechanisms can be used for tasks that are not subject to high sensitivity, such as medians, or complex tasks for which no unit split equivalent is available, such as stochastic gradient descent \cite{AbadiCGMMT016}. This allows for tighter analysis when using low sensitivity queries and  opens up the vast literature of differentially private techniques for use alongside unit splitting.

\subsection{Answer GroupBy Aggregation Queries Using Unit Splitting.}\label{sec:unit-splitting-groupby}

In addition to aggregations such as sums and counts, unit splitting also supports conditional analysis such as filters and group-by. Filters and group-by can be applied directly on the conditional attributes after unit splitting, since those attributes are duplicated across all splits. This allows for individual analysis for each group. \par 
Lemma~\ref{lem:PRzCDP-par} allows a practitioner to apply different splitting thresholds for each group in order to better serve the needs of each group. For example, consider if we added Technology as an additional industry in Table~\ref{tab:ex_splitting}. Technology firms have significantly higher average pay than agricultural establishments and as such have a significantly higher payroll. In such cases, a different set of splitting thresholds may be necessary for technology firms to avoid extremely high privacy loss. \par 
In cases where the group-by keyset is unknown, or is sparse in the domain of the attribute one can use a private partition selection algorithm, after the unit splitting process. Since the data has been split prior to partition selection, large records are instead split into many small records and are as such more likely to be discovered. We demonstrate an example of using all three techniques; group-by, multiple splitting thresholds, and split partition selection.
\begin{example}\label{ex:SQL}
    Consider taking two sums over the Employee column and Payroll column of Table~\ref{tab:ex_splitting}(a), grouped by the values of the Industry column. In addition, consider the following splitting  thresholds for each industry. For agriculture and retail, use the previous splitting threshold of (Employees: 50, Payroll: \$5,000,000). For mining, apply the splitting threshold of (Employees: 50, Payroll: \$10,000,000). \par 
    
    To compute the sum, we need to complete three steps. First apply unit splitting to each industry with their own splitting threshold. This bounds the sensitivity of both the Employees and Payroll column, however this bound is now different for each industry. Then use some privacy budget $\rho_1$ to use a partition selection technique \cite{desfontaines_differentially_2022} to find the keyset for the group-by. Since each industry is split prior to the partition selection, those with few but large records still have a high probability to be in the resulting keyset. Then for each category in the Industry column, we take the sum over employees using the Gaussian mechanism with sensitivity $50$  and privacy budget $\rho_2$. This sensitivity remains the same since the employee splitting threshold is the same for all industries.
   For sum over payroll, use the Gaussian mechanism with sensitivity $5,000,000$ for agriculture and retail and use sensitivity $10,000,000$ for mining. For both, we will use the same privacy budget $\rho_3$. \par  Since the partitions over the industry column form disjoint subsets of the dataset by Lemma~\ref{lem:PRzCDP-par}, the sum over Employees and Payroll satisfy $\rho_2|A(r)|^2$ and $\rho_3 |A(r)|^2$-PRzCDP respectively where $|A(r)|$ denotes the number of splits for each record for their respective splitting thresholds. We can combine all of these policy functions together using Lemma~\ref{lem:PRzCDP-seq} to get that the overall mechanism satisfies $P$-PRzCDP with a policy function $P(r) =  (\rho_1 + \rho_2 +\rho_3)|A(r)|^2$. For records 1 and 2, the final privacy loss is $9(\rho_1 + \rho_2 + \rho_3)$ since they are each split into 3 rows. Record 3 incurs $4(\rho_1 +\rho_2 + \rho_3)$ privacy loss since it is split into 2 rows under the new splitting threshold, and record 4 only incurs $\rho_1 +\rho_2 + \rho_3$ privacy loss since it is not split under the new splitting thresholds.
\end{example}
Since these were disjoint sections of the database, we could apply different splitting thresholds to each disjoint section and still satisfy $P$-PRzCDP. In this case, the new policy function would be a piece-wise function, giving a different functional form for each industry.
\par
Due to the initial unit splitting, the partition selection step has a high probability of selecting each of the populated industries, even if those industries are populated by few but large establishments. This allows practitioners to properly analyze heavily skewed data where much of the data is compressed into relatively few records, a phenomenon which often happens in economic or population statistics.

\section{Experiments}
\label{sec:experiments}
In this section, we empirically demonstrate the effectiveness of PRzCDP when applied to skewed data. In Section \ref{sec:exp_negative} we demonstrate how the high bias and error from global zCDP results in an unacceptable trade-off between privacy and utility. Then, in Section \ref{sec:exp_skewed}, we focus on univariate queries to show how PRzCDP can be an effective alternative to zCDP with modest reductions in privacy. Finally, in Section \ref{sec:exp_cbp}, we demonstrate the methodology on our motivating use case, CBP. 

\subsection{Setup and Datasets}
\label{sec:exp_setup}
 For each experiment, we answer queries using either zCDP or PRzCDP according to the Gaussian mechanism (Def. \ref{def:gaussian_mech}) with clamping-enforced sensitivity $\Delta$, noise variance $\sigma^2$, and privacy loss $\rho = \frac{\Delta^2}{2\sigma^2}$, or unit splitting pre-processing (Algorithm \ref{alg:unit-splitting}) with additive Gaussian noise with variance $\sigma^2$ according to different splitting functions $T$. We use the following metrics throughout. Where contextually appropriate, we abuse notation and only include the relevant arguments.

    \textbf{Policy loss}: for $r \in D$, we have
    \begin{equation}
        \mathrm{PolicyLoss}(P, r) \triangleq P(r)
    \end{equation}
    Note that in practice, we cannot release $\mathrm{PolicyLoss}(P, r)$, only the functional form $P(\cdot)$. \par 
    \textbf{Realized loss}: for a record $r$, the realized dataset $D$, and mechanism $M$, 
    \begin{align}
        &\mathrm{RealizedLoss}(M, D, r) \\
        &\triangleq \sup_{\alpha \in (1, \infty)} \frac{d_\alpha(M(D) || M(D \setminus \{r \}))}{\alpha}
    \end{align}
    By construction, when $M$ satisfies $P$-PRzCDP, 
    \begin{equation}
        \mathrm{PolicyLoss}(P, r) \geq \mathrm{RealizedLoss}(M, D, r)
    \end{equation}
    for all $r \in D$. Again, in practice, we cannot release \linebreak $\mathrm{RealizedLoss}(M, D, r)$ nor its functional form, as it depends on the realized dataset $D$. \par 
    \textbf{Query relative error}: for a dataset $D$, query $M(D)$, and non-private answer $S(D)$, define
    \begin{align}
        &\mathrm{QueryRelErr}(M, S, D, \gamma) \\
        &\triangleq \min\left\{ v \in \mathbb{R}^+ \mid \p\left( \frac{|M(D) - S(D)|}{S(D)} \geq v \right) \leq 1 - \gamma \right\}
    \end{align} \par 
    \textbf{Absolute relative error (ARE)}: for a given output $M(D)$ and non-private answer $S(D)$, define
    \begin{align}
        \mathrm{ARE}(M(D), S(D)) \triangleq \frac{|M(D) - S(D)|}{|S(D)|}
    \end{align}\par 
    
     \textbf{Policy minimum}: for a policy function $P$, we will use \linebreak $\mathrm{PolicyMin}(P) = \min_{r \in \mathcal{T}} P(r)$ to denote the smallest policy loss associated with one record. For unit splitting algorithms, this can be interpreted as the policy loss for records unaffected by unit splitting.

Our experiments are run on three different datasets: a simulated dataset (SIM), the National Agricultural Statistical Service Cattle Inventory Survey (CIS), and a County Business Patterns (CBP) synthetic dataset. \par

The first dataset is simulated data for which we know the precise data generating distribution. This allows us to illustrate the kinds of heavy-tailed behavior that our methodology can better accommodate, as opposed to global DP. 

We simulate two heavy-tailed variables in $[1, \infty)$ with tail index parameters $\alpha$, where smaller values of $\alpha$ correspond to heavier tails. The simulated variables are listed in Table \ref{tab:sim_ht_vars}; note that both variables HT1 and HT2 have infinite variance. Our goal is to answer sum queries of HT1 and HT2 grouped by CatIX. 

\begin{table}[!htbp]
    \centering
    \begin{tabular}{|c|c|c|}
        \hline
        Name & Domain & Distribution \\
        \hline
        $\mathrm{CatIX}$ & $\{ 1, \dots, 1000\}$ & $\mathrm{Categorical}(\phi)$ \\
        \hline
        $\mathrm{HT1}$ & $[1, \infty)$ & $\mathrm{Pareto}(1, 1.2)$ \\
        \hline
        $\mathrm{HT2}$ & $[1, \infty)$ & $\mathrm{Pareto}(1, 1.5)$ \\
        \hline
    \end{tabular}
    \caption{Simulated heavy-tailed variables}
    \label{tab:sim_ht_vars}
\end{table}

The second dataset is from the U.S. Department of Agriculture (USDA)'s Cattle Inventory survey (CIS), managed by the National Agricultural Statistical Service (NASS) \cite{agcensus}. Our records consist of county-level survey records of total cattle inventory and average pastureland rent cost (in dollars per acre). Each county geography is contained hierarchically within an agricultural district (AD), itself contained within a particular state. We will treat these records as our privacy units, since the individual farm-level records are not publicly available. Still, a small number of records contribute to the majority of the total cattle inventory in any particular state or AD, making the methodology applicable. We consider the queries in Table \ref{tab:cis_queries} under different privacy loss allocations and splitting thresholds shown in the figures. The $P$-PRzCDP queries are answered using Algorithm \ref{alg:unit-splitting} and the $\rho$-zCDP queries are answered using the Gaussian mechanism from Definition \ref{def:gaussian_mech}. 

\begin{table}[!htbp]
    \centering
    \begin{tabular}{|c|c|c|}
        \hline
        Geographies & Query & Formalism \\
        \hline 
        State & SUM(CattleInventory) & $P$-PRzCDP \\
        \hline
        State & AVG(PastureRent) & $\rho$-zCDP \\
        \hline 
        State x AD & SUM(CattleInventory) & $P$-PRzCDP \\
        \hline
        State x AD & AVG(PastureRent) & $\rho$-zCDP \\
        \hline
    \end{tabular}
    \caption{CIS Queries}
    \label{tab:cis_queries}
\end{table}

The final dataset is a proposed use case involving summations over skewed data: the County Business Patterns (CBP) dataset, published by the U.S. Census Bureau. We use synthetic data provided by the U.S. Census Bureau to demonstrate PRzCDP methodology on these synthetic records. Each row in our tabular data represents an ``establishment", or one separate unit of a business; ``firms" represent one or more establishments that operate as a single business venture. Our goal is to release the following information about groups of establishments: 
\begin{itemize}
    \item ESTAB: A count of the number of establishments.
    \item PAYANN: A sum of annual payrolls of establishments.
    \item PAYQTR1: A sum of first quarter payrolls of establishments.
    \item EMP: A sum over employee size of establishments.
\end{itemize}
The groups of establishments correspond to different geographic areas (such as counties, ZIP codes, or congressional distributions) and different industry classifications using the North American Industry Classification System (NAICS) codes (such as finance, real estate, agriculture, etc.). For the purposes of this experiment, we consider the subset of all county-level queries at every possible NAICS classification level with no cross-tabulations; moreover, we limit our evaluations to only those queries with 100 or more establishments.

\begin{table}[!htbp]
    \centering
    \begin{tabular}{|c|c|c|}
        \hline
        Geographies & Query & Formalism \\
        \hline 
        County x NAICS* & COUNT(ESTAB) & $\rho$-zCDP \\
        \hline
        County x NAICS* & SUM(EMP) & $P$-PRzCDP \\
        \hline
        County x NAICS* & SUM(PAYANN) & $P$-PRzCDP \\
        \hline
        County x NAICS* & SUM(PAYQTR1) & $P$-PRzCDP \\
        \hline
    \end{tabular}
    \caption{CBP Queries}
    \label{tab:cbp_queries}
\end{table}

\subsection{Global zCDP on Skewed data}
\label{sec:exp_negative}

First, we show how a theoretical analysis of the $\rho$-zCDP Gaussian mechanism for sums on heavy-tailed random variables fails to yield reasonable trade-offs between privacy and utility. Specific to our simulation study, we consider the theoretical mean-square error (MSE) of estimators truncated with high probability from heavy tails. In Figure \ref{fig:sim_mse_global}, we fix $n=1000$ and plot the theoretical mean-square error (MSE) over the sum's expected value as a measure of "noise-to-signal" on the y-axis. We show how this ratio varies with different sensitivities $\Delta$ on the x-axis, privacy losses $\rho$, and tail weights $\alpha$. Note that in every case, privacy-preserving noise exceeds the sum's expected value by orders of magnitude. For each configuration, we additionally calculate the optimal $\Delta$ for given $\alpha$ and $\rho$ values which minimize the ratio (shown as the vertical red lines in the subplots). As expected, the optimal $\Delta$ value to minimize MSE increases as $\alpha$ decreases and as $\rho$ increases; however, even at these optimal $\Delta$ values, the errors are prohibitively large. Recall from \cite{bun_concentrated_2016} that Gaussian noise is tight for zCDP summation queries, meaning any $\rho$-zCDP mechanism requires noise with variance at least $\Omega(\Delta/\rho^2)$, i.e. as a function of the volume of the sum query space. So even for modest tail weights, the cost of ensuring each record lies in a bounded domain makes it near-impossible to simultaneously maintain modest global privacy losses and MSE guarantees. 

\label{sec:exp_sim}
\begin{figure}
    \centering
    \includegraphics[width=.45\textwidth]{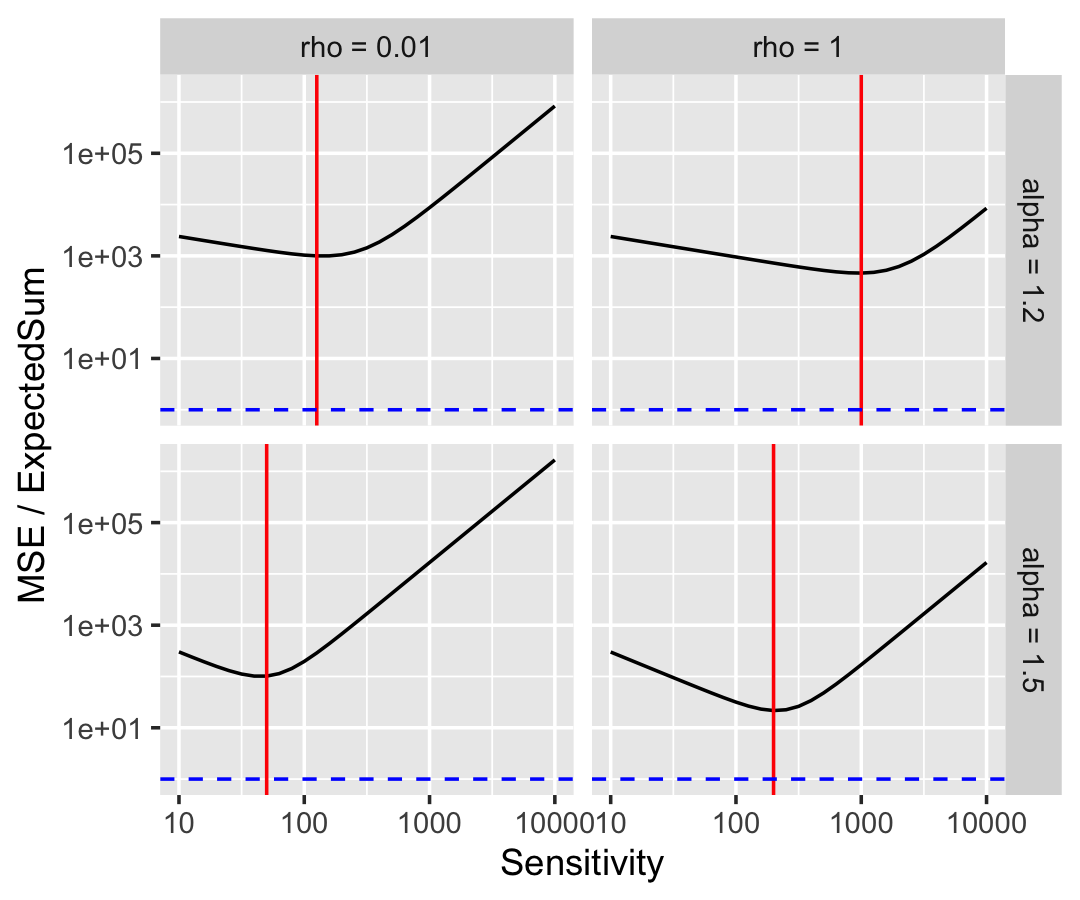}
    \caption{Theoretical MSE over the expected query value for global $\rho$-zCDP mechanisms with different global sensitivities $\Delta$, privacy loss budgets $\rho$, and different tail parameters $\alpha$ for $n = 1000$. Optimal $\Delta$ for minimizing MSE given $\rho$ and $\alpha$ shown in red. Blue dashed line at 1, for reference.}
    \label{fig:sim_mse_global}
\end{figure}

\subsection{PRzCDP privacy-utility trade-offs with univariate splitting}
\label{sec:exp_skewed}
Next, we show how using Algorithm \ref{alg:unit-splitting} in conjunction with the Gaussian mechanism offers significant improvements to utility with a cost to privacy loss that only affects a small number of units. In Figure \ref{fig:sim_cis_priv_vs_util}, we use our proposed method to answer queries on the workloads with different univariate unit splitting thresholds (STs) for one variable of interest (HT1 for SIM, CattleInventory for CIS). The left set of subplots show the workload AREs (y-axis) at different STs; as ST decreases, the proportion of records that are split (for which policy loss is greater than $\rho$) increases, shown on the x-axis. To simplify, utility increases (y-axis distributions shift downward) as privacy loss increases (x-axis boxplots shift to the right). The plots show that ARE improves significantly, while policy loss for most records remains the same as the global zCDP counterpart. For example, at $\rho = 1$ for the simulated data, we can achieve a median 10\% ARE across queries while ensuring less than 1\% of records have policy loss greater than $\rho = 1$. 

On the right-hand side of Figure \ref{fig:sim_cis_priv_vs_util}, we show a more detailed view of the policy loss functions by visualizing their empirical CDFs: namely, for any one unit splitting configuration, what proportion of records (y-axis) have policy losses less than a particular value (x-axis)? We show this for different STs and $\rho$. As $\rho$ increases, the CDFs shift to the right, as expected since less noise injection increases policy loss uniformly across records. Larger STs correspond to more conservative unit splitting schemes, ensuring that a greater proportion of records have the smallest possible policy loss. As ST decreases, the policy loss grows more rapidly for larger units, which are split more frequently. These plots demonstrate how $\rho$ toggles the privacy-utility trade-off for all records, whereas ST toggles how fast the policy loss grows as records become more skewed.

\begin{figure*}[t]
    \centering
    a) Simulated Pareto data \\
    \includegraphics[width=.49\textwidth]{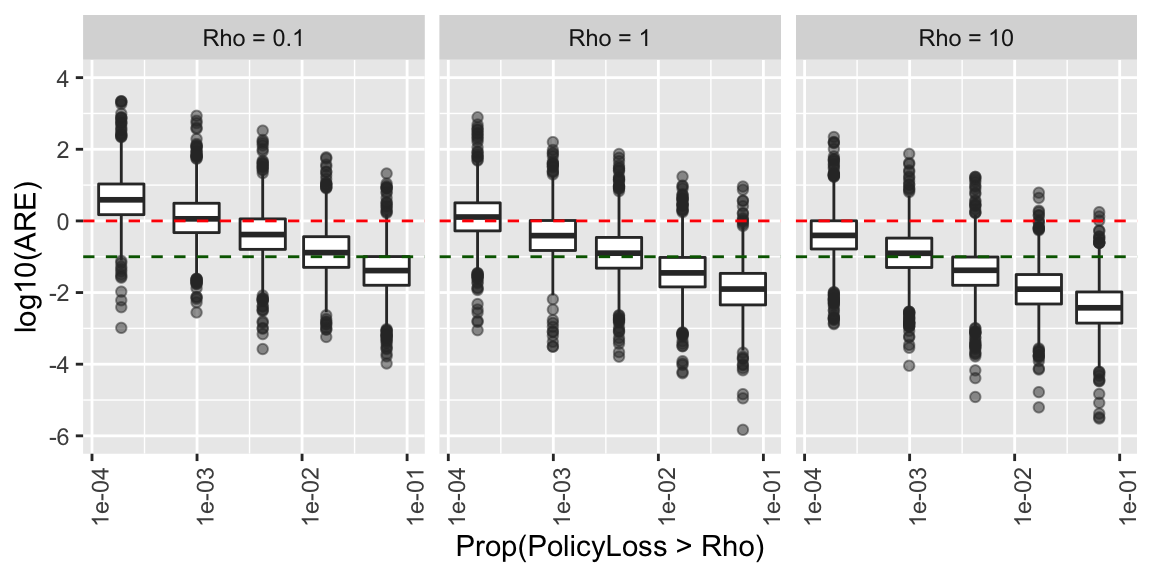}
    \includegraphics[width=.49\textwidth]{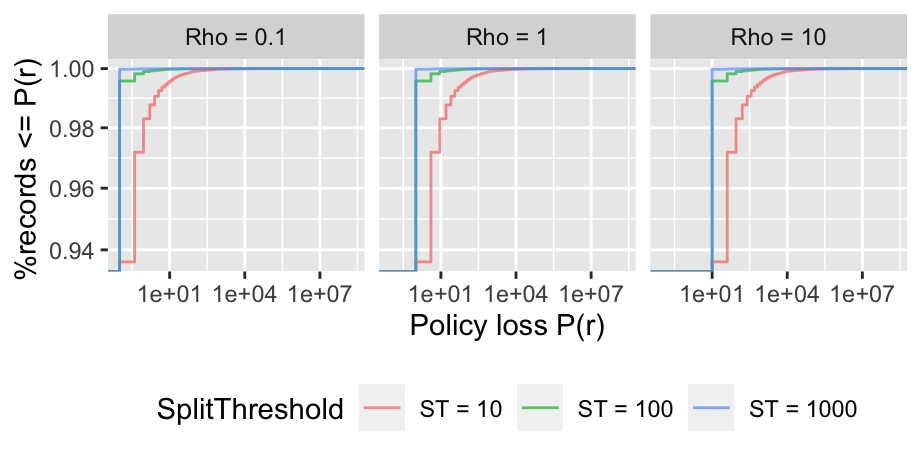}
    b) USDA Cattle Inventory Study \\
    \includegraphics[width=.49\textwidth]{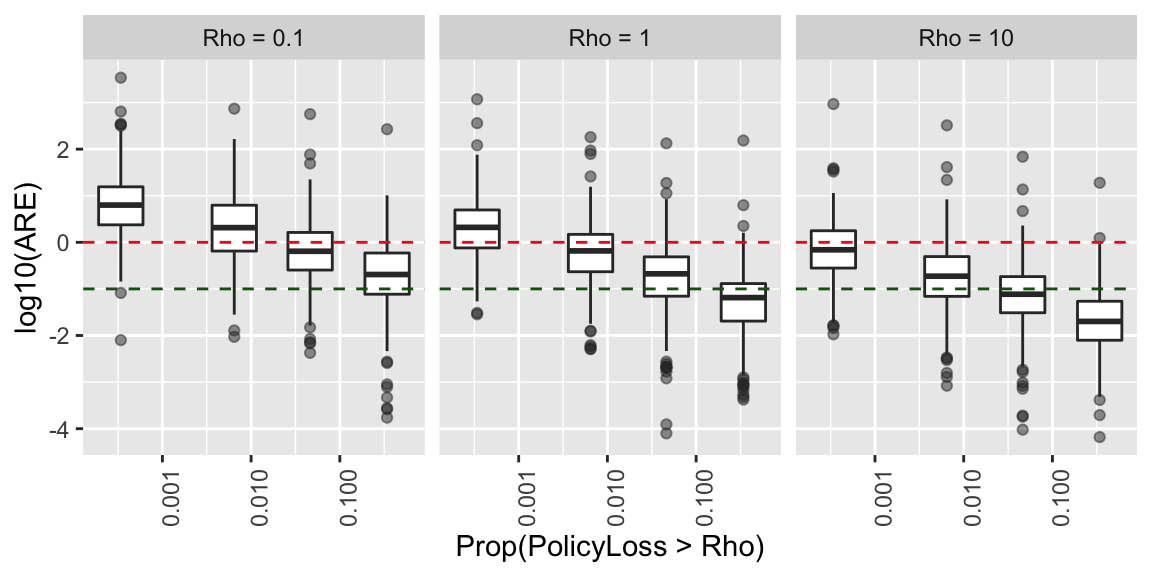}
    \includegraphics[width=.49\textwidth]{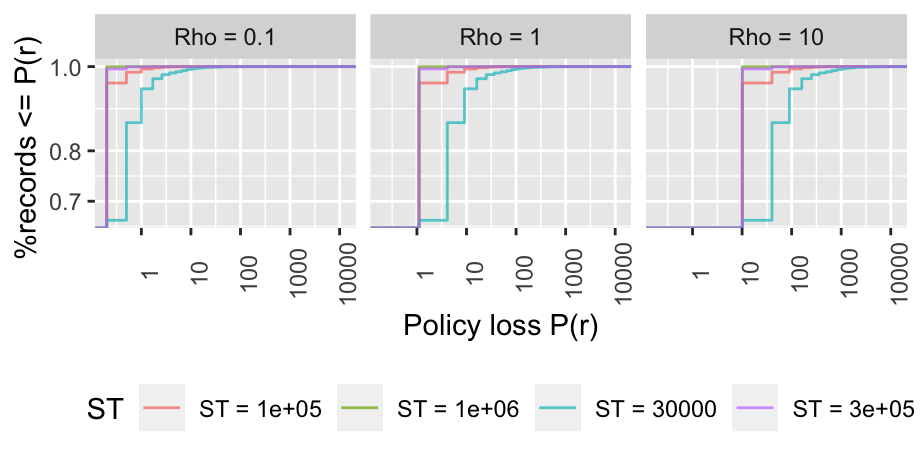}
    \caption{(Left) distribution of ARE over workload queries (y-axis) by proportion of records with policy loss greater than $\rho$ (x-axis). (Right) Empirical CDFs of policy loss, i.e. proportion of observed records (y-axis) with policy loss bounded by $P(r)$ (x-axis). Columnar subplots show different levels of minimum policy loss $\rho$. Red line represents 100\% ARE and green line represents 10\% ARE.}
    \label{fig:sim_cis_priv_vs_util}
\end{figure*}

\subsection{End-to-end example: County Business Patterns dataset}
\label{sec:exp_cbp}
We now turn towards a more complex, realistic application of our methodology to CBP. The query workload is described in Table \ref{tab:cbp_queries}; answering these queries requires leveraging more features of our proposed framework. First, we consider multivariate unit splitting as a function of multiple attributes per record. Second, we combine zCDP with PRzCDP queries. Third, we use both sequential and parallel composition simultaneously to answer queries about the full workload.

\begin{figure*}[t]
    \centering
    \includegraphics[width=.48\textwidth]{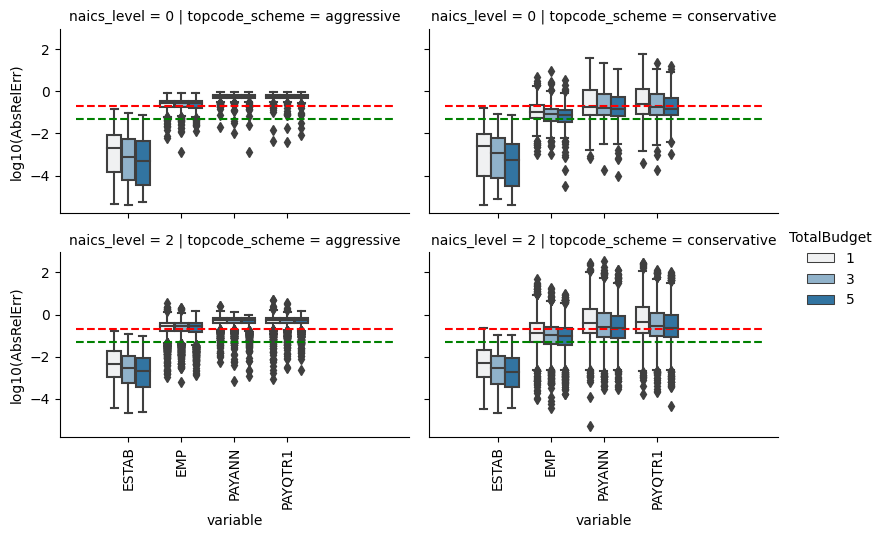}
    \includegraphics[width=.48\textwidth]{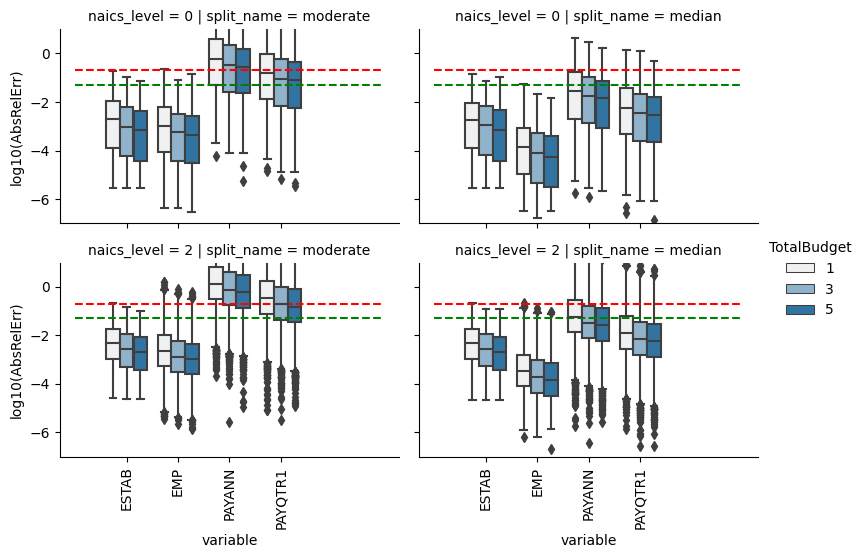}
    \caption{AREs for the CBP query workload using topcoding and zCDP (left) versus using unit splitting (right) for different NAICS levels (rows) and splitting schemes (columns). Red line represents 20\% ARE and green line represents 5\% ARE. }
    \label{fig:cbp_ares}
\end{figure*}

We consider two different algorithmic approaches for answering the CBP query workload. First, we consider zCDP mechanisms using sensitivities defined by three possible sets of ``top-codes" in Table \ref{tab:top_codes}, which we name ``Conservative", ``Moderate", and ``Aggressive", in decreasing  order. Second, we consider PRzCDP based on splitting schemes listed in Table \ref{tab:splitting_thresholds}, which we name ``Conservative", ``Moderate", and ``Median". The three schemes are listed in increasing split cardinality order. 

\begin{table}[!htpb]
    \centering
    \begin{tabular}{|c|c|c|}
        \hline
        Top-code scheme name & Attribute & Value \\
        \hline
        \multirow{3}{*}{Conservative} & EMP & $10^3$ \\ \cline{2-3}
         & PAYANN & $10^5$ \\ \cline{2-3}
         & PAYQTR1 & $2.5*10^4$ \\
        \hline
        \multirow{3}{*}{Moderate} & EMP & $3*10^2$ \\ \cline{2-3}
         & PAYANN & $10^4$ \\ \cline{2-3}
         & PAYQTR1 & $2.5*10^3$ \\
        \hline
        \multirow{3}{*}{Aggressive} & EMP & $10^2$ \\ \cline{2-3}
         & PAYANN & $10^3$ \\ \cline{2-3}
         & PAYQTR1 & $2.5*10^2$ \\
        \hline
        
    \end{tabular}
    \caption{Top-code scheme description.}
    \label{tab:top_codes}
\end{table}

\begin{table}[!htpb]
    \centering
    \begin{tabular}{|c|c|c|c|}
        \hline
        SchemeName & SplitAttribute & SplitThreshold & PctScore \\
        \hline
        \multirow{3}{*}{Conservative} & PAYANN & 10000 & 99\% \\ \cline{2-4}
         & PAYQTR1 & 2500 & 99\% \\ \cline{2-4}
         & EMP & 100 & 97\% \\ 
        \hline
        \multirow{3}{*}{Moderate} & PAYANN & 500 & 79\% \\ \cline{2-4}
         & PAYQTR1 & 125 & 80\% \\ \cline{2-4}
         & EMP & 5 & 66\% \\ 
        \hline
        \multirow{3}{*}{Median} & PAYANN & 104 & 50\% \\ \cline{2-4}
         & PAYQTR1 & 24 & 50\% \\ \cline{2-4}
         & EMP & 2 & 47\% \\
        \hline 
    \end{tabular}
    \caption{Establishment splitting thresholds for three different splitting schemes and their associated percentiles of score in the CBP synthetic data. }
    \label{tab:splitting_thresholds}
\end{table}

In Figure \ref{fig:cbp_ares} we plot the ARE of each query for the top-coding algorithm and the establishment splitting algorithm, respectively. The results are aggregated by attribute, total privacy loss budget, NAICS level, and algorithmic configuration. The green and red dashed lines mark the 5\% and 20\% ARE thresholds, respectively, representing example fitness-for-use goals. First, we expect the relative errors for counting attributes (i.e. ESTAB and FIRM) to have the same distribution for either algorithm, as they are unaffected by establishment splitting.  However, when we look at the magnitude attributes (EMP, PAYANN, and PAYQTR1), none of the top-coding schemes on the left come close to providing reasonable AREs, since the majority of the box plot masses for these queries are above the dashed green line. Alternatively, on the right, we see that establishment splitting provides far smaller relative errors, even for more granular queries at finer NAICS levels (although the relative errors increase as the NAICS level increases, as expected).

\begin{figure*}[t]
    \centering
    \includegraphics[width=.7 \textwidth]{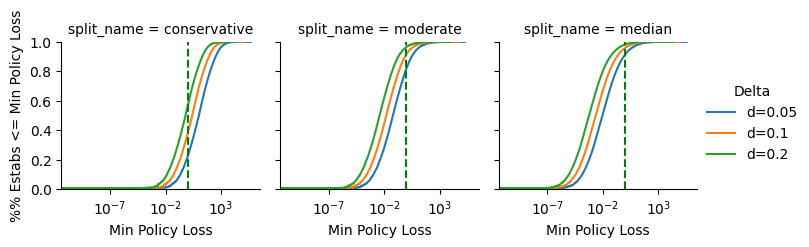}
    \caption{Theoretical minimum policy function CDFs to achieve different fitness-for-use goals on 95\% of the COUNTY by NAICS code query workload. The green dashed line represents the total unsplit privacy loss budget of 1.}
    \label{fig:ffu_split}
\end{figure*}

\begin{figure}
    \centering
    \includegraphics[width=.48\textwidth]{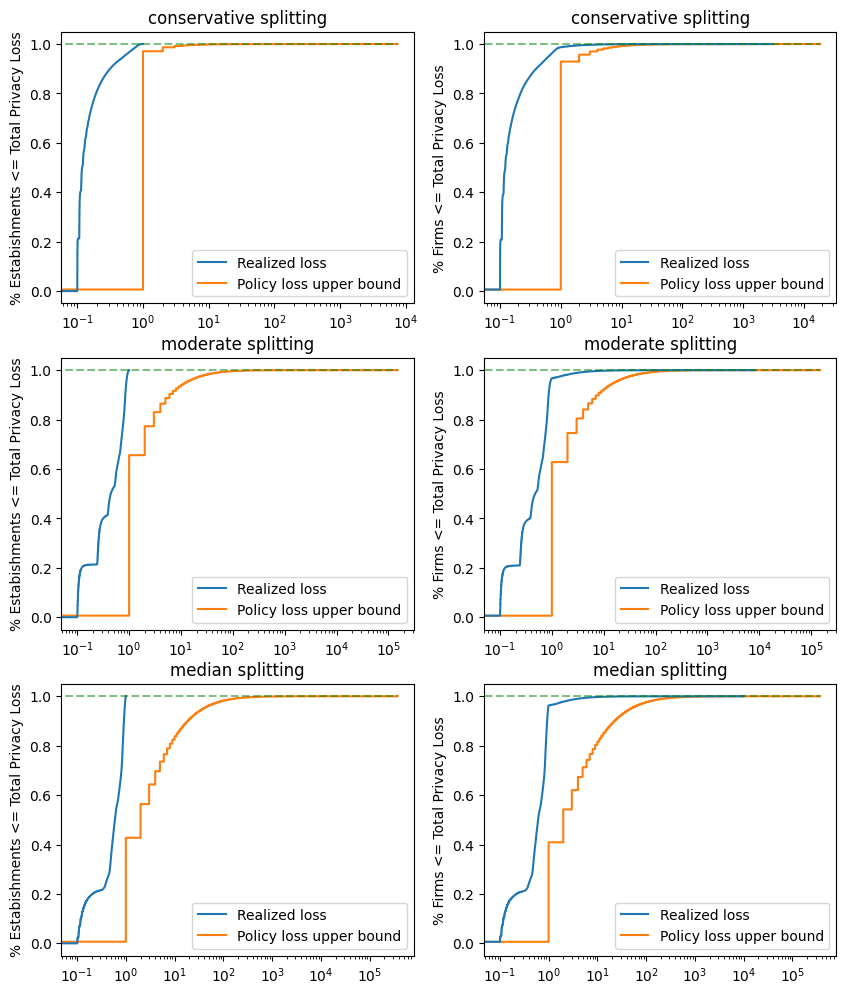}
    \caption{CBP policy losses grouped by establishment (left) and firm (right)}
    \label{fig:cbp_losses}
\end{figure}

Similarly, in Figure \ref{fig:cbp_losses}, we plot the policy loss CDFs for each splitting scheme, which can be interpreted similarly to the policy loss CDFs in Figure \ref{fig:sim_cis_priv_vs_util} with a few differences worth highlighting. First, by PRzCDP's group composition properties, we can extend the results from the left column of establishment subplots to the right column of firm subplots. Since the majority of establishments in the synthetic data have a unique ID, the plots look very similar; however, the firm-level CDFs sit slightly below the establishment level CDFs. This demonstrates how establishment-level guarantees are conferred to firms. Second, we additionally calculate the \emph{realized} privacy losses for each establishment and firm. Specifically, this calculates the log-max divergence between the establishment splitting outputs using the specific CBP synthetic dataset with and without the establishment (or firm) of interest. By construction, the realized loss is always less than the policy loss, so the blue CDF line will always be to the left of the orange policy loss line. First, we observe that for the majority of establishments, the realized privacy loss is significantly lower than the policy upper bound. Moreover, because this gap depends on the number of queries answered in the workload, we can reasonably expect this gap to be larger when considering the entire CBP workload, not just county-level queries. Next, we observe that, as the splitting thresholds decrease, the gap between the realized privacy loss and the policy upper bound decreases. 

Finally, instead of considering the realized relative errors, we ask: what is the \emph{smallest} possible policy function which ensures we reach a particular fitness-for-use goal? Specifically, we calculate the smallest policy loss function for each entity where we assume that for each query in the county workload, we have at theoretical query relative error of less than $\delta$ with probability at least 95\%. We plot the implied policy loss CDFs in Figure \ref{fig:ffu_split}. As expected, we require smaller policy losses for the majority of establishments as the splitting thresholds get smaller and smaller, since we are incurring larger privacy losses for larger establishments. Additionally, as expected, when $\delta$ increases, the distribution of the minimum policy loss subsequently decreases (i.e., the CDFs are shifted to the left). All this demonstrates that with these splitting schemes, fitness-for-use goals are more feasible than under the traditional top-coding assessment.

\section{Conclusion}
To summarize, we introduced PRzCDP to transparently encode dependencies between per-record privacy loss and confidential records. This relaxation of traditional DP notions helps answer SQL-style queries over skewed data where approaches like zCDP may fail to offer reasonable privacy-utility trade-offs. By making the policy function public, we offer a new way of describing privacy loss in cases where a small minority of records pose exorbitant privacy risks that aren't representative of the entire dataset. Such policy functions are particularly useful when the unit of privacy analysis is not an individual person, but a group of people in a business establishment or other organization where there may be different social privacy expectations for large versus small groups. We additionally offer a way of indirectly setting the policy function through unit splitting, a pre-processing step that composes with DP algorithms to provide PRzCDP guarantees by construction. Our experiments applying this technique to simulated data, cattle inventory data, and business pattern data demonstrate how PRzCDP can better answer realistic SQL-style query workloads on skewed data without relying on zCDP's worst-case analysis.
Future work beyond the scope of the article could more formally characterize the semantic guarantees offered by PRzCDP. Techniques like unit splitting intrinsically leak more information about confidential records when records are split with finer granularity. Understanding the kinds of queries that could be leveraged to learn confidential information via the policy function require further investigation. While this paper focuses on static data dissemination settings, such extensions would be helpful for using PRzCDP in more general (possibly interactive) query systems. Similarly, future work could explore different techniques for choosing how privacy loss scales with record values. While we considered quadratic dependence between record values and policy loss, using $\epsilon$-DP style semantics could yield linear dependence instead. Alternatively, additional pre-processing and post-processing steps could enable new ways of scaling this dependence. 
\begin{acks}
    We would like to thank Margaret Beckom, Anthony Caruso, William Davie Jr, Ian Schmutte, and Brian Finley from the U.S. Census Bureau and Zach Terner from MITRE for their valuable insight and feedback.
\end{acks}
\bibliographystyle{ACM-Reference-Format}
\bibliography{references}

\appendix

\section{Omitted Proofs}

\subsection{Proof of Lemma \ref{lem:PRzCDP-seq}}
Fix $D \in \mathcal{D}$ and let $D' \in \mathcal{D}$ such that $D \ominus D' = \{ r' \}$. Let $Y_1, Y_2$ be the outputs from $M_1(D)$ and $M_2(D)$, respectively, and let $B_1 \in \mathcal{F}_{Y_1}$ and $B_2 \in \mathcal{F}_{Y_2}$. Then:
\begin{align*}
    d_\alpha\left(\left( M_1(D), M_2(D) \right) \mid \left( M_1(D'), M_2(D') \right) \right) \\ = d_\alpha(M_1(D) \mid M_1(D')) + d_\alpha(M_2(D) \mid M_2(D')) \\
    \leq \alpha P_1(r) + \alpha P_2(r) \\
    = \alpha \left( P_1(r) + P_2(r) \right)
\end{align*}

\subsection{Proof of Lemma \ref{lem:PRzCDP-par}}
Without loss of generality, let $D,D' \in \mathcal{D}$ such that $D \ominus D' \triangleq \{ r'\} \subseteq C_1$. Then:
\begin{align*}
D_\alpha(M(D) || M(D')) = \sum_{j=1}^J D_\alpha(M_j(D \cap C_j) || M_j(D' \cap C_j)) \\
\leq \alpha P(r') +  \left[ \sum_{j = 2}^J D_\alpha(M_j(D \cap C_j) || M_j(D' \cap C_j)) \right]\\
= \alpha P(r')
\end{align*}

\subsection{Proof of Lemma \ref{lem:PRzCDP-group}}

We will prove the result by strong induction. Define the induction hypothesis for $m \in [J]$:
$$
d_\alpha(M(D_0) \mid M(D_m)) \leq \alpha \left( m \sum_{j=1}^m P(r_j) \right)
$$
The hypothesis holds by definition for $m = 1$. Next, suppose it holds for $m = J - 1$. From Lemma 5.2 in \cite{bun_concentrated_2016}, we have, for any $k, \alpha \in (1, \infty)$:
\begin{equation}
d_\alpha(P || Q) \leq \frac{k\alpha}{k\alpha - 1} d_{\frac{k\alpha - 1}{k - 1}}(P || R) + d_{k\alpha}(R || Q)
\end{equation}
Then using the lemma where $k = J$:
\begin{align*}
    &d_\alpha(M(D_0) || M(D_J)) \\ 
    &\leq \frac{J\alpha}{J\alpha - 1} d_{\frac{J\alpha -1}{J-1}}(M(D_0) || M(D_{J-1})) + d_{J\alpha}(M(D_{J-1}) || M(D_J)) \\
    &\leq \frac{J\alpha}{J\alpha - 1} \left( \frac{J\alpha - 1}{J - 1} \right) \left( (J-1) \sum_{j=1}^{J-1} P(r_j) \right) + J\alpha P(r_J) \\
    &\leq \alpha \left( J\sum_{j=1}^J P(r_j) \right)
\end{align*}
By induction, the result holds for $m = J$. Note that the expression holds without loss of generality for any permutation of $\{ r_1, \dots, r_J \}$. 

\subsection{Proof of Lemma \ref{lem:advanced-PRzCDP-par}}

From Lemma 5.2 in \cite{bun_concentrated_2016}, we have, for $k, \alpha \in (1, \infty)$
\begin{equation}
d_\alpha(P || Q) \leq \frac{k\alpha}{k\alpha - 1} d_{\frac{k\alpha - 1}{k - 1}}(P || R) + d_{k\alpha}(R || Q)
\end{equation}
In iterated form, we have:
\begin{equation}
d_{k^m \alpha}(P || Q) \leq \frac{k^{m+1}\alpha}{k^{m+1}\alpha - 1} d_{\frac{k^{m+1}\alpha - 1}{k - 1}}(P || R) + d_{k^{m+1} \alpha}(R || Q)  
\end{equation}
Using the iterated form, we have:
\begin{align*}
&d_\alpha(M(D_0) || M(D_J)) \\
&\leq \frac{k\alpha}{k\alpha - 1} d_{\frac{k\alpha - 1}{k - 1}}(M(D_0) || M(D_1)) + d_{k\alpha}(M(D_1) || M(D_J)) \\
&\leq \frac{k\alpha}{k\alpha - 1} d_{\frac{k\alpha - 1}{k - 1}}(M(D_0) || M(D_1)) + \\ & \frac{k^2 \alpha}{k^2 \alpha - 1} d_{\frac{k^2\alpha - 1}{k - 1}}(M(D_1) || M(D_2)) + d_{k^2\alpha}(M(D_2) || M(D_J)) \\
&\leq \cdots \\
&\leq \sum_{j=1}^J \frac{k^j \alpha}{k^j\alpha - 1} d_{\frac{k^j \alpha - 1}{k - 1}}(M(D_{j-1})||M(D_j))) \\
&\leq \sum_{j=1}^J \frac{k^j \alpha}{k^j\alpha - 1} \left( \frac{k^j \alpha - 1}{k - 1} \right) P(r_j) \\
&\leq \alpha \sum_{j=1}^J \frac{k^j}{k - 1} P(r_j)
\end{align*}
Since $\frac{k^j}{k-1}$ is an increasing function in $j$ for $j \in [J]$ and $k \in (1, \infty)$, this expression is minimized by the sequence of order statistics defined above by finding the minima over all $k \in (1, \infty)$:
\begin{align*}
&d_\alpha(M(D_0) || M(D_J)) \leq \alpha \sum_{j=1}^J \frac{k^j}{k - 1} P(r_{(j)}) \quad \forall k \in (1, \infty) \\
&\implies d_\alpha(M(D_0) || M(D_J)) \leq \alpha \inf_{k \in (1, \infty)} \sum_{j=1}^J \frac{k^j}{k - 1} P(r_{(j)})
\end{align*}
As an example where the bound improves upon the result from Lemma \ref{lem:PRzCDP-group}, let $J=2$. Then:
\begin{align*}
&d_\alpha(M(D_0) || M(D_2)) \\
&\leq \frac{k\alpha}{k\alpha - 1} d_{\frac{k\alpha - 1}{k - 1}}(M(D_0) || M(D_1)) + d_{k\alpha}(M(D_1) || M(D_2)) \\
&\leq \frac{k\alpha}{k -1} P(r_{(1)}) + k\alpha P(r_{(2)})
\end{align*}
Finding the minimum solution within $k \in (1, \infty)$ yields:
$$
\frac{d}{dk} \left[ \frac{k}{k-1} P(r_{(1)}) + k P(r_{(2)}) \right] = -\frac{P(r_{(1)})}{(k -1)^2} + P(r_{(2)}) = 0 \implies $$ $$   k^* = \sqrt{\frac{P(r_{(1)})}{P(r_{(2)})}} + 1
$$
This implies:
\begin{align*}
d_\alpha(M(D_0) || M(D_2)) \leq \alpha\left[ \frac{k^*}{k^*-1} P(r_{(1)}) + k^* \alpha P(r_{(2)}) \right] \\
= \alpha\left[ \frac{\sqrt{\frac{P(r_{(1)})}{P(r_{(2)})}} + 1}{\sqrt{\frac{P(r_{(1)})}{P(r_{(2)})}}} P(r_{(1)}) + \left( \sqrt{\frac{P(r_{(1)})}{P(r_{(2)})}} + 1 \right) P(r_{(2)}) \right] \\
= \alpha \left[ P(r_{(1)}) + P(r_{(2)}) + 2 \sqrt{P(r_{(1)})P(r_{(2)})}\right]
\end{align*}
Therefore, the result will demonstrate an improvement over Lemma \ref{lem:PRzCDP-group} whenever:
\begin{align*}
2\sqrt{P(r_{(1)}) P(r_{(2)})} < P(r_{(1)}) + P(r_{(2)}) \\ \iff (P(r_{(1)}) - P(r_{(2)}))^2 > 0 \\
\iff P(r_{(1)}) > P(r_{(2)})
\end{align*}
Therefore, such improvement occurs only when the policy yields non-constant privacy loss across records, which agrees with the global zCDP result. Moreover, this result \textit{only} holds when the records are considered to be added or removed in the most favorable order. In the least favorable order, in which the order statistics are reversed, the composition is weaker, i.e.:
$$
\frac{}{} \alpha \sum_{j=1}^J \frac{k^j}{k-1} P(r_{(J+1-j)}) > \alpha J \sum_{j=1}^J P(r_j)
$$
In other words, the simple composition offered by the previous theorem is tighter among order-agnostic group privacy guarantees.\par 
Similar to previous definitions, PRzCDP is also robust to post-processing. That is, if a mechanism $M$ satisfies PRzCDP, then any function applied directly to the output of $M$ also satisfies PRzCDP.

\subsection{Proof of Lemma \ref{lem:pre-processing}}
Let $D$ and $D'$ be neighboring datasets and, without loss of generality, let $A(D') \setminus A(D) = \{ s_1, \dots s_{|A(r)|} \} \subseteq \mathcal{T}^*$. By the group-privacy properties of $\rho$-zCDP:
\begin{align}
    D_\alpha(M(A(D)) || M(A(D'))) \leq \alpha \left( |A(r)|^2 \rho \right)
\end{align}

\end{document}